\documentclass[graybox]{svmult}

\usepackage{mathptmx}       
\usepackage{helvet}         
\usepackage{courier}        
\usepackage{type1cm}        
\usepackage{makeidx}         
\usepackage{graphicx}        
\usepackage{multicol}        
\usepackage[bottom]{footmisc}

\usepackage{amsmath,amsfonts}

\begin{document}

\title*{An exceptional symmetry algebra for the 3D Dirac--Dunkl operator}

\author{Alexis Langlois-R\'emillard and Roy Oste} 

\institute{Alexis Langlois-R\'emillard and Roy Oste \at Department of Applied Mathematics, Computer Science and Statistics, Faculty of Sciences, Ghent University, Krijgslaan 281-S9, 9000 Gent, Belgium \newline \email{Alexis.LangloisRemillard@UGent.be}; \email{Roy.Oste@UGent.be} }
\maketitle

\abstract{We initiate the study of an algebra of symmetries for the 3D Dirac--Dunkl operator associated with the Weyl group of the exceptional root system $G_2$.  
For this symmetry algebra, we give both an abstract definition and an explicit realisation. 
We then construct ladder operators, using an intermediate result we prove for the Dirac--Dunkl symmetry algebra associated with arbitrary finite reflection group acting on a three-dimensional space.}

\section{Introduction}
\label{sec:1}

In the present paper, we initiate the study of an algebra of symmetries for the Dirac--Dunkl operator associated with the exceptional root system $G_2$.  The latter is primarily known from the classification of simple Lie algebras. The associated Lie group and algebra continue to spark interest, see for instance the recent paper of Dobrev~\cite{Dobrev} and references therein. 
Our purpose is related instead to the action of the Weyl group associated with $G_2$ on a (two-dimensional subspace of a) three-dimensional space. 
Though $G_2$ is indeed a root system of rank 2, the arising symmetry algebra associated with three-dimensional space portrays interesting non-trivial relations, which are not present when considering the two-dimensional analogue. 

We will briefly recall how the symmetry algebra in question arises.
For a finite reflection group $W$ acting on a finite dimensional vector space, there exists a rational Cherednik algebra~(RCA)~\cite{Etingof} that can be viewed as a deformation of the algebra of polynomial differential operators on the vector space. An explicit realisation is given by means of differential-difference operators called Dunkl operators~\cite{Dunkl}. A generalisation of the Dirac operator is defined abstractly inside the tensor product of the RCA and a Clifford algebra, or explicitly by using Dunkl operators in lieu of partial derivatives in the ordinary definition of the Dirac operator. 

In this way, the Dirac--Dunkl operator squares to a Dunkl version of the Laplace operator whose invariance is restricted to the group $W$ as opposed to the full orthogonal invariance of its classical counterpart. Moreover, together with its dual partner, the Dirac--Dunkl operator generates a Lie superalgebra isomorphic to $\mathfrak{osp}(1|2)$. The latter's (super)centraliser inside the tensor product of RCA and Clifford algebra gives an algebra of symmetries (super)commuting with the Dirac--Dunkl operator. Structurally it can be seen as a deformation of the orthogonal Lie algebra representing total angular momentum in the non-deformed case. 

In previous work~\cite{Oste}, explicit expressions for the elements of the symmetry algebra and the generated algebraic structure were determined for arbitrary finite reflection group. Subsequently, the study was specialised to the $A_2$ root system with Coxeter group $S_3$ acting on a three-dimensional Euclidean space~\cite{Oste2}. In this case it was possible to classify all irreducible representations and give conditions for when they are unitarisable. An important tool was the construction of ladder operators. 

A natural follow-up question is whether this approach extends to settings with other reflection groups. The existence of ladder operators will in general depend on the root system under consideration. One of our aims is to work out in detail the conditions for their existence. The full analysis goes beyond the scope of this contribution; here we will already present some preliminary results pertaining to three-dimensional spaces and focus in particular on the exceptional root system $G_2$, embedded herein. 

In section~\ref{sec:2} the required definitions of the exceptional root system $G_2$ and Dirac--Dunkl operator are introduced and we present the symmetry algebra both abstractly and as an explicit realisation. 
In section~\ref{sec:3}, we prove an intermediate result for arbitrary root system in $\mathbb{R}^3$ and show that this leads to the existence of ladder operators for the symmetry algebra associated with $G_2$.

\section{An exceptional symmetry algebra}
\label{sec:2}

We consider the Euclidean space $\mathbb{R}^3$ with coordinates $x_1,x_2,x_3$. 
The 2-dimensional root system $G_2$ is realised in a plane and is generated by two simple roots $\alpha_1 =(0,1,-1)$ and $\alpha_2 = (1,-2,1)$. The Coxeter group linked to $G_2$ is the dihedral group $D_{2\cdot 6}$ that we will present by: $D_{12} = \langle \sigma_1,\sigma_2 \mid \sigma_1^2=\sigma_2^2=(\sigma_1\sigma_2)^6 = (\sigma_2\sigma_1)^6=1\rangle$ with the reflections $\sigma_1$ connected to the short root $\alpha_1$, and $\sigma_2$ to the long root $\alpha_2$. Their actions on $\mathbb R^3$ are expressed matricially by:
\begin{equation}
	\sigma_1 = \begin{pmatrix}
		1 & 0 & 0\\
		0 & 0 & 1\\
		0 & 1 & 0
	\end{pmatrix}, 
	\quad
	\sigma_2 = \begin{pmatrix}
		2/3 & 2/3 &-1/3\\
		2/3 &-1/3 & 2/3\\
		-1/3 & 2/3 & 2/3
	\end{pmatrix}.
\end{equation}

A set of positive roots is given by 
\begin{equation}
	\begin{split}
		R_+ &= \left\{\alpha_1 = (0,1,-1), \alpha_2 = (1,-2,1), \alpha_3 = (1,-1,0), \right.\\
		&\qquad {}\left.\alpha_4 =(1,1,-2) , \alpha_5 =(1,0,-1) , \alpha_6 = (2,-1,-1)\right\}.	
	\end{split}
\end{equation}
To each root $\alpha_i$, a reflection $\sigma_i$ is paired. The reflections have the following decompositions in terms of the simple reflections $\sigma_1,\sigma_2$:
\begin{align}
	\sigma_3 &= \sigma_2\sigma_1\sigma_2, & \sigma_4 &= \sigma_1\sigma_2\sigma_1, &  \sigma_5 &= \sigma_1\sigma_2\sigma_1\sigma_2\sigma_1, & \sigma_6 &= \sigma_2\sigma_1\sigma_2\sigma_1\sigma_2.
\end{align}

We introduce a $D_{12}$-invariant weight function $\kappa : G_2 \to \mathbb C$, which is defined by two complex numbers $\kappa_1$ and $\kappa_2$ linked respectively to the short and long roots. With this, it is possible to define Dunkl operators~\cite{Dunkl} for the root system $G_2$; for example the one associated with the coordinate $x_2$ is given by
\begin{equation}
	\begin{split}
		\mathcal{D}_2 &= \frac{\partial}{\partial x_2} + \kappa_1\left( \frac{1-\sigma_1}{ x_2-x_3} + \frac{1-\sigma_3}{x_1-x_2}\right)\\
		&\quad {} + \kappa_2\left( -2\frac{1-\sigma_2}{ x_1-2x_2+x_3} + \frac{1-\sigma_4}{ x_1+x_2-x_3} - \frac{1-\sigma_6}{ 2x_1-x_2-x_3}\right),
	\end{split}
\end{equation}
while $\mathcal{D}_1$ and $\mathcal{D}_3$ are defined similarly.

Next, we consider the Clifford algebra with three anticommuting generators $e_1,e_2,e_3$ that all square to $\varepsilon \in \{+1,-1\}$. The Dirac--Dunkl operator associated with our embedding of $G_2$ in $\mathbb{R}^3$ is realised explicitly by $\mathcal{D} = \mathcal{D}_1 e_1 + \mathcal{D}_2 e_2 + \mathcal{D}_3 e_3$. Together with its dual partner $x_1 e_1 + x_2 e_2 + x_3 e_3$, it generates a realisation of $\mathfrak{osp}(1|2)$. 
For ease of notation, we shall not make explicit mention of the tensor product, trusting the reader to add it whenever Clifford elements $e_i$ are involved.

The elements of the symmetry algebra were obtained in previous work~\cite{Oste} (that they indeed generate the full centraliser is the subject of~\cite{Oste4}) and we will go over them now. First, we need a double cover of the Weyl group $D_{12}$. 
The orthogonal group $O(3)$ has two non-isomorphic double covers. These correspond to the two choices of $\varepsilon$ in the definition of the Clifford algebra~\cite{Morris}. 
For either choice of $\varepsilon$, we obtain a double cover $\widetilde{D}_{12}^{\varepsilon}$ by viewing $D_{12}$ as a subgroup of the orthogonal group $O(3)$, through the pullback of the projection of the $\mathrm{Pin}^{\varepsilon}(3)$ double cover onto $O(3)$. In this way, we obtain 
the $\widetilde{D}_{12}^{\varepsilon}$ elements (together with their additive inverses):
\begin{align*}
	\widetilde{\sigma}_1 &= \frac{\sigma_1 (e_2-e_3)}{\sqrt{2}}, & \widetilde{\sigma}_3 &= \frac{\sigma_3(e_1-e_2)}{\sqrt{2}}, & \widetilde{\sigma}_5 &= \frac{\sigma_5(e_1-e_3)}{\sqrt{2}}, \\
	\widetilde{\sigma}_2 &= \frac{\sigma_2(e_1-2e_2+e_3)}{\sqrt{6}}, & \widetilde{\sigma}_4 &= \frac{\sigma_4(e_1+e_2-2e_3)}{\sqrt{6}}, & \widetilde{\sigma}_6 &= \frac{\sigma_6(2e_1-e_2-e_3)}{\sqrt{6}}.
\end{align*}

Note that the group relations depend on the choice of $\varepsilon$. By direct computation we find 
$\widetilde{D}_{12}^{\varepsilon} = \langle \widetilde{\sigma}_1, \widetilde{\sigma}_2\mid \widetilde{\sigma}_1^2 = \widetilde{\sigma}_2^2 = \varepsilon, (\widetilde{\sigma}_1\widetilde{\sigma}_2)^6 = (\widetilde{\sigma}_2\widetilde{\sigma}_1)^6 = -1\rangle$, which also follows from~\cite[Thm 4.2]{Morris}.
The order of this group is 24, and for $\varepsilon =+1$ it is again a dihedral group, while for $\varepsilon =-1$ it is a dicyclic group. 
Regardless of the choice of $\varepsilon$, all elements of $\widetilde{D}_{12}^{\varepsilon}$ will supercommute with the Dunkl-Dirac operator when taking into account the $\mathbb{Z}_2$-grading inherited from the Clifford algebra. Both $\mathcal{D}$ and $\pm\sigma_i$ are odd elements with respect to this grading, so they will in fact anticommute. 
In the following, we will use the standard notation for anticommutator $(\{-,-\})$ and commutator $([-,-])$.

Furthermore, there are three analogues of the total angular momentum operators that commute with the Dirac operator: $O_{12}, O_{23}, O_{13}$. Classically (non-Dunkl) they generate a realisation of the orthogonal Lie algebra $\mathfrak{so}(3)$, though here it will be a deformation thereof. An explicit realisation is given by
\begin{equation}\label{eq:deftwoindexsym}
	O_{ij} = L_{ij} + \varepsilon e_ie_j/2 + O_ie_j - O_je_i,
\end{equation}
where $L_{ij} = x_i\mathcal{D}_j-x_j\mathcal{D}_i$ is a Dunkl analogue of angular momentum, 
and for ease of notation we denote some specific linear combinations of elements of $\widetilde{D}_{12}^{\varepsilon}$ as follows:
\begin{equation}\label{eq:defoneindexsym}
	\begin{aligned}
		O_1 &= \kappa_1( \widetilde{\sigma}_3 + \widetilde{\sigma}_5 ) + \kappa_2(\widetilde{\sigma}_2 + \widetilde{\sigma}_4 +2\widetilde{\sigma}_6),\\
		O_2 &= \kappa_1(\widetilde{\sigma}_1  - \widetilde{\sigma}_3) + \kappa_2(-2 \widetilde{\sigma}_2+\widetilde{\sigma}_4 -\widetilde{\sigma}_6),\\
		O_3 &= \kappa_1(-\widetilde{\sigma}_1  -\widetilde{\sigma}_5) +\kappa_2(\widetilde{\sigma}_2 -2\widetilde{\sigma}_4 -\widetilde{\sigma}_6) .
	\end{aligned}
\end{equation}
It is immediate to see that the sum $O_1 + O_2 + O_3 = 0$. Moreover, we will denote $\mathcal E = \left[O_1, O_2\right]$, and by direct but slightly tedious computations, we can also see that $\left[O_2,O_3\right] =\mathcal E =-\left[O_1,O_3\right]$. 
From the realisation~\eqref{eq:deftwoindexsym}, it is clear that $O_{ij} = - O_{ji}$, and it is convenient to abide by this convention also when defining the algebra elements abstractly.

The interaction of the two simple reflections $\widetilde{\sigma}_1$ and $\widetilde{\sigma}_2$ with the two-index symmetries of equation \eqref{eq:deftwoindexsym} are given by:
\begin{equation}\label{eq:acttwoindexsym}
	\begin{aligned}
	\widetilde{\sigma}_1 O_{12} &= O_{13} \widetilde{\sigma}_1, & \widetilde{\sigma}_2O_{12} &= (-2/3 O_{12} + 2/3 O_{13} + 1/3O_{23})\widetilde{\sigma}_2,\\
	\widetilde{\sigma}_1 O_{13} &= O_{12} \widetilde{\sigma}_1, & \widetilde{\sigma}_2O_{13} &= (2/3 O_{12} + 1/3 O_{13} + 2/3O_{23})\widetilde{\sigma}_2,\\
	\widetilde{\sigma}_1 O_{23} &= -O_{23} \widetilde{\sigma}_1, & \widetilde{\sigma}_2O_{23} &= (1/3 O_{12} + 2/3 O_{13} - 2/3O_{23})\widetilde{\sigma}_2;
	\end{aligned}
\end{equation}
from which the entire action of $\widetilde{D}_{12}^{\varepsilon}$ follows. 

The final generator of our symmetry algebra is a central element $O_{123}$, of which an explicit realisation is given by
\begin{equation}\label{eq:defthreeindexsym}
	O_{123} = \varepsilon e_1e_2e_3 + O_1 e_2e_3 - O_2 e_1e_3 + O_3 e_1e_2 + L_{12}e_3 - L_{13}e_2+L_{23}e_1.
\end{equation}
As a consequence of the relations in the general case, see \cite[Thm~3.12]{Oste} or \cite[eq.~(1.7)]{Oste2}, the two-index symmetries~\eqref{eq:deftwoindexsym} respect 
\begin{equation}\label{eq:twoindexsym}
	\begin{aligned}
		\left[O_{13}, O_{12}\right] &= O_{23} + 2O_{123}O_1 + \mathcal E; \\
		[O_{23}, O_{12}] &= -O_{13} + 2O_{123}O_2 + \mathcal E; \\
		[O_{23}, O_{13}] &= O_{12} + 2O_{123}O_3 + \mathcal E. 
	\end{aligned}
\end{equation}
These relations can be proved specifically for the $G_2$ case, in a similar manner as was done for $S_3$~\cite{Oste3}. 

In the right-hand sides of \eqref{eq:twoindexsym} appear the linear combinations of elements of $\widetilde{D}_{12}^{\varepsilon}$ given by~\eqref{eq:defoneindexsym} and $\mathcal E$. When the deformation parameters $\kappa_1$, $\kappa_2$ are chosen to be zero, these all vanish and the relations~\eqref{eq:twoindexsym} reduce to those of the orthogonal Lie algebra $\mathfrak{so}(3)$.

\section{Ladder operators}
\label{sec:3}

The result we prove next holds for arbitrary root system in $\mathbb{R}^3$. Hereto, one should use the appropriate definitions for $O_1,O_2,O_3$ as given in~\cite[eq.~(3.8) and Ex. 4.2]{Oste} and the relations analogous to~\eqref{eq:twoindexsym} given by \cite[eq.~(1.7)]{Oste2}. 
What we obtain in this way are not yet the desired ladder operators, though we will show that they do lead to ladder operators for the $G_2$ case at hand.

\begin{proposition}\label{prop:commutationwithO0andal}
	Let $\omega = e^{2i\pi /3}$ and consider the following linear combinations: 
	\begin{equation}\label{eq:eleO}
		\begin{aligned}
			O_0 &= -i/\sqrt{3}(O_{12} + O_{23} - O_{13}),\\
			O_+ &= -i\sqrt{2/3}(O_{12} + \omega O_{23} - \omega^2 O_{13}),\\
			O_- &= -i\sqrt{2/ 3} (O_{12} + \omega^2 O_{23} - \omega O_{13}).
		\end{aligned}
	\end{equation}
	Denoting $\omega^{+} = \omega$ and $\omega^{-} = \omega^{2}$, they satisfy
		\begin{equation}\label{eq:commutationwithO0andal}
		\begin{aligned}
			\left[O_0, O_{\pm}\right] &= \pm O_{\pm} \mp i\sqrt{2/3}  (2O_{123}(O_3 + \omega^{\pm} O_1 + \omega^{\mp}O_2) \\ 
			& \quad {} + \left[ O_1,O_2\right] + \omega^{\pm}\left[ O_2,O_3\right] + \omega^{\mp}\left[ O_3,O_1\right]); \\
			\left[O_+,O_-\right] &= 2 O_0  - 2i/\sqrt{3} (2O_{123}(O_1+O_2+O_3) \\ 
			& \quad {} + \left[O_1,O_2\right] +\left[ O_2,O_3\right] + \left[ O_3,O_1\right]).
		\end{aligned}
	\end{equation}
\end{proposition}
\begin{proof}	
	Using the definitions~\eqref{eq:eleO} and 
grouping the terms appropriately we obtain
	\begin{align*}
		\left[O_0, O_{\pm}\right] &= -\sqrt{2}/ 3 \left( (1-\omega^{\pm}) \left[ O_{23},O_{12}\right]\right.\\
		&\quad {} + \left.(\omega^{\mp}-1)\left[ O_{12},O_{31}\right] + (\omega^{\pm} - \omega^{\mp}) \left[ O_{31},O_{23}\right]\right).
	\intertext{Noticing that $(\omega^{\pm} - \omega^{\mp}) = \pm i\sqrt{3}$, and $(1-\omega^{\pm}) = 3/2 \mp i\sqrt{3}/2 = \pm i\sqrt{3} \omega^{\mp}$, and $(\omega^{\mp} -1) = -3/2 \mp i\sqrt{3}/2 = \pm i\sqrt{3}\omega^{\pm}$, and applying~\cite[eq.~(1.7)]{Oste2} results in}
		&= \mp i\sqrt{2}/\sqrt{3} \big(  \omega^{\mp} (O_{31} + \{ O_{123},O_2\} + \left[O_3, O_1\right])\\
		&\qquad \qquad {} + \omega^{\pm} (O_{23} + \{ O_{123},O_1\} + \left[O_2, O_3\right])\\
		&\qquad \qquad {} +   O_{12} + \{ O_{123},O_3\} + \left[O_1, O_2\right]\big),
		\end{align*}	
		and finally using again the definition~\eqref{eq:eleO} one arrives at the desired expression.

	In the same manner for the second equation, we find
	\begin{align*}
		\left[ O_+, O_-\right] &= -2/ 3 (\omega - \omega^2)\left(\left[O_{23},O_{12}\right] + \left[ O_{12},O_{31}\right] + \left[O_{31},O_{23}\right]\right)\\
		&= -2i/\sqrt{3} \left( O_{31} + \{ O_{123},O_2\} + \left[ O_3,O_1\right]\right. + O_{23} + \{O_{123},O_1\} + \left[ O_{2},O_1\right]\\
		&\qquad \qquad \quad {} + \left.O_{12} + \{ O_{123},O_1 \} + \left[ O_1,O_2\right]\right)\\
		&= 2O_0 - 2i/\sqrt{3} \left(\{O_{123},O_1+O_2+O_3\} + \left[O_1,O_2\right] +\left[ O_2,O_3\right] + \left[ O_3,O_1\right]\right).
	\end{align*}
	As $O_{123}$ is central, this proves the second equality.\smartqed\qed
\end{proof}

When the root system satisfies some specific properties, we can use the previous result to obtain ladder operators.

\begin{proposition}\label{prop:ladderopg2}
	For the root system $G_2$, the elements $O_{0}$, $O_+$ and $O_-$ satisfy
	\begin{equation}\label{eq:formofcommutO0andalG2}
		\begin{aligned}
			\left[O_0, O_{\pm}\right] &= \pm O_{\pm} \mp 2i\sqrt{2/3} \,O_{123} \left( O_3 + \omega^{\pm} O_1 + \omega^{\mp}O_2\right); \\
			\left[O_+,O_-\right] &= 2 O_0  - 2i\sqrt{3} \mathcal E .
		\end{aligned}
	\end{equation}
	Moreover, the quadratic elements $K_{\pm} = 1/2\{O_0,O_{\pm}\} $ fulfill the ladder operator relations $[O_0, K_{\pm}] = \pm K_{\pm}$.
\end{proposition}
\begin{proof}	\smartqed
	Starting from the relations~\eqref{eq:commutationwithO0andal}, we can use $1 + \omega + \omega^2 = 0$, and  $O_1 + O_2 + O_3 = 0$, while $\left[O_1, O_2\right]=\left[O_2,O_3\right] =\left[O_3,O_1\right]=\mathcal E $, to arrive at~\eqref{eq:formofcommutO0andalG2}.

In addition, we have
$
		\left[ O_0,K_{\pm}\right] = 1/2\left[ O_0,\left\{ O_0,O_{\pm}\right\}\right] = 1/2\left\{ O_0,\left[ O_0,O_{\pm}\right]\right\}
$.
By the first relation~\eqref{eq:formofcommutO0andalG2}, this becomes
\begin{equation*}
		\left[ O_0,K_{\pm}\right] =\pm 1/2\left\{ O_0,O_{\pm}\right\} \mp i\sqrt{2/3}\left\{O_{0}, O_{123}(O_3 + \omega^{\pm} O_1 + \omega^{\mp}O_{2})\right\} = \pm K_{\pm}.
	\end{equation*}
In the last step we used the fact that $O_{123}$ is central, and that all elements of $\widetilde{D}_{12}^{\varepsilon}$ anticommute with $O_0$, which is clear from the action~\eqref{eq:acttwoindexsym}.\smartqed\qed
\end{proof}

These ladder operators can now be used in the study of the representation theory of the symmetry algebra in a similar vein as was done in the $S_3$ case~\cite{Oste2}, which we aim to do in future work.  In addition, we will investigate the construction of ladder operators for other reflection groups.

\begin{acknowledgement}
We wish to thank Hendrik De Bie and Joris Van der Jeugt for helpful discussions and support.
This research was supported in part by EOS Research Project number 30889451.
ALR also holds a scholarship from the Fonds de recherche du Qu\'ebec -- Nature et technologies number 270527. RO was supported by the Joint Research Project KP-06-N28/6 of the Bulgarian National Science Fund and by a postdoctoral fellowship, fundamental research, of the Research Foundation -- Flanders (FWO), number 12Z9920N. 
This support is gratefully acknowledged. 
\end{acknowledgement}

\end{document}